\documentclass[sigconf]{acmart}

\usepackage{booktabs}
\usepackage{multirow}
\usepackage{multicol}
\usepackage{enumerate}
\usepackage{subfigure}
\usepackage{lipsum}
\usepackage{algorithm}
\usepackage[noend]{algorithmic}
\usepackage{threeparttable}
\usepackage{tikz}
\usepackage{graphicx}
\usetikzlibrary{arrows, shapes, positioning}
\usetikzlibrary{decorations.markings}
\tikzstyle arrowstyle=[draw]
\tikzstyle directed=[postaction={decorate,decoration={markings,
		mark=at position 0 with {\node[circle,fill=black,inner sep=0.3ex] {};},
		mark=at position 0.5 with {\arrow[scale=1.2,arrowstyle]{stealth};},
		mark=at position 1 with {\node[circle,fill=black,inner sep=0.3ex] {};},
}}]

\AtBeginDocument{%
  }

\setcopyright{acmlicensed}
\copyrightyear{2018}
\acmYear{2018}
\acmDOI{XXXXXXX.XXXXXXX}

\acmConference[Conference acronym 'XX]{Make sure to enter the correct
  conference title from your rights confirmation email}{June 03--05,
  2018}{Woodstock, NY}
\acmISBN{978-1-4503-XXXX-X/18/06}

\begin{document}

\title{An efficient recursive decomposition algorithm for  undirected graphs}

\author{Pei Heng,Yi Sun and Jianhua Guo}
\authornote{The first two authors contributed equally to this research.}

\renewcommand{\shortauthors}{Trovato et al.}

\begin{abstract}
  Graph decomposition is a versatile and powerful tool used to tackle complex systems on a large scale, playing a crucial role in various fields like graph theory, artificial intelligence, and statistics. In this study, we explore the relationship between atoms (i.e., maximal prime subgraphs) and convex hulls in undirected graphs. Building on this connection, we introduce a recursive decomposition algorithm with polynomial complexity for undirected graphs that doesn't rely on triangulation. Our experimental findings demonstrate that this algorithm operates significantly faster than current state-of-the-art methods. 
\end{abstract}

\ccsdesc[300]{General and reference~General literature}
\ccsdesc[500]{Theory of computation~Graph algorithms analysis}

\keywords{Undirected graph, Decomposition, Atom, Convex hull}

\received{20 February 2007}
\received[revised]{12 March 2009}
\received[accepted]{5 June 2009}

\maketitle

\section{Introduction}
An effective strategy for handling complex graphs is the divide-and-conquer approach. This method involves breaking down the graph into smaller subgraphs using specific separators, solving individual problems on each subgraph, and then combining the local solutions to derive a global solution. This technique is widely used in various tasks such as minimal triangulation \cite{berry2010introduction}, probabilistic reasoning 
\cite{lauritzen1988local, olesen2002maximal, jensen1990}, constraint satisfaction \cite{dechter1988tree}, and structural learning \cite{geng2005decomposition,zhang2020learning}.

Divide-and-conquer strategies in graph theory often involve using separators to construct separation trees or employing decompositions to build decomposition trees. This paper focuses on the latter approach, which decomposes undirected graphs using clique minimal separators. This technique is also known as simplicial decomposition \cite{diestel1990graph} or atom decomposition  \cite{berry2010introduction}, terms used interchangeably here. Graph decomposition techniques trace back to Gavril's algorithm \cite{gavril1977algorithms} in 1977 for efficiently dividing graphs into subgraphs, initially aimed at solving problems like minimum coloring and maximum cliques.

Whitesides \cite{Whitesides1981AnAF} introduced a method to discover clique separators with a computational complexity of $O(n^3)$, where $n$ denotes the number of nodes in the network. Tarjan \cite{tarjan1985decomposition} expanded on this work by proposing an $O(mn)$ algorithm, where $m$ represents the number of edges, for decomposing arbitrary graphs using clique separators. Tarjan also highlighted that Whitesides's algorithm operates in 
$O(n^3m)$ time. Leimer \cite{leimer1993optimal} building on Tarjan's work \cite{tarjan1985decomposition}, developed an enhanced decomposition method with $O(nm)$ complexity. This approach accurately decomposes graphs into their maximal prime subgraphs by first employing the LEX-M algorithm \cite{rose1976algorithmic} to establish a minimal triangulation ordering. Kratsch and Spinrad \cite{kratsch2006minimal} proposed an algorithm with $O(n^{2.69})$ complexity for determining minimal triangulation orderings, and showed that clique separator decomposition can also be achieved within the same time bound. 

Berry et al. \cite{berry2014organizing} introduced further optimizations to reduce the time required for decomposition using minimal triangulation orderings. Coudert and Ducoffe \cite{coudert2018revisiting} demonstrated that decomposing a graph via clique separators is at least as challenging as triangle detection, and they showed that a given undirected graph can be decomposed in $O(n^\alpha \log n) = \widetilde O(n^{2.3729})$ \footnote{The $\widetilde O$ notation ignores logarithmic factors.} time. As of current research, all widely recognized atom decomposition algorithms for general graphs are rooted in the concept of minimal triangulation for graphs. This approach is well-documented in related works such as those by \cite{berry2010introduction, berry2014organizing,coudert2018revisiting}.

The typical execution framework of these algorithms, as outlined in \cite{berry2014organizing}, follows these steps:

\begin{itemize}
    \item \textbf{Compute the Minimal Triangulation Ordering}: First, determine a minimal triangulation ordering for a graph. This involves ordering the vertices such that the resulting graph is chordal (also known as a triangulated graph).
    \item \textbf{Find Clique Minimal Separators}: Using the computed minimal triangulation ordering, identify clique minimal separators in the graph. These are subsets of vertices whose removal leaves connected components that are themselves cliques (complete subgraphs).
    \item \textbf{Recursively Decompose the Graph}: Utilize the identified clique minimal separators to recursively decompose the graph into smaller subgraphs. This process continues until each resulting subgraph is fully decomposed into its prime components or atoms.
\end{itemize}

This structured approach leverages the properties of chordal graphs and the efficient computation of minimal triangulations to facilitate effective decomposition using clique separators. It represents a significant advancement in graph theory, enabling complex graph problems to be tackled through a divide-and-conquer strategy built on these foundational concepts. Finding the minimum triangulation ordering of a graph is known to be NP-hard \cite{yannakakis1981computing}, which leads existing algorithms to focus on achieving minimal triangulation. Currently, two primary methods exist for determining this ordering in undirected graphs: one involves variable elimination \cite{ohtsuki1976minimal,rose1976algorithmic,berry2004maximum,kratsch2006minimal}, while the other utilizes saturating relative minimal separators \cite{parra1997characterizations,berry2006wide,heggernes2006minimal}. The optimal complexity of the minimal triangulation algorithm is 
$o(n^{2.376})$, as proposed by \cite{heggernes2005computing}. For further details on minimal triangulation algorithms, see the survey by \cite{heggernes2006minimal}. 


Kratsch and Spinrad \cite{kratsch2006between} note that despite knowing the minimal triangulation ordering of a graph, finding clique separators remains challenging. Recently, Coudert and Ducoffe \cite{coudert2018revisiting} demonstrated that clique minimal separators can be computed in $o(n^{2.3729})$ time using the minimal triangulation ordering of a graph. However, although these algorithms theoretically break the $O(n^3)$ barrier, the practical improvements in execution efficiency are limited. This results in daunting runtime complexities for clique separator decomposition algorithms, especially when applied to large networks with tens of thousands of nodes \cite{kratsch2006between,crescenzi2013computing}.


In this paper, we reveal that atoms can be identified by examining convex hulls that encompass specific vertices. This insight inspires a recursive approach to decomposing large graphs, in which atoms are recursively identified by considering convex hulls containing closures of designated variables.  Our method addresses existing limitations and significantly enhances decomposition efficiency. The main contributions of this work encompass both theoretical advancements and algorithmic innovations:


\begin{itemize}
    \item \textbf{Theoretical Contribution:} We show that the convex hull containing the closure of a node with the smallest maximum cardinality search (MCS) ordering number forms an atom of the graph.
    \item \textbf{Algorithmic Contribution:} We devise an efficient recursive decomposition algorithm for graphs. This algorithm identifies atoms by recursively searching for convex hulls that encompass specific variables. Furthermore, we have extended this algorithm to a parallel version, enhancing efficiency by strategically selecting the MCS ordering.
    \item \textbf{Experimental Validation:} Empirical experiments on real networks are conducted to compare and analyze the performance of our algorithm with traditional approaches. The results show that our method significantly outperforms traditional decomposition algorithms in terms of efficiency.
\end{itemize}

The rest of this paper is structured as follows: Section~\ref{sec-2} introduces essential graph notation and terminology. Section~\ref{sec-3} explores the theoretical connection between atoms and convex hulls, from which we derive our recursive decomposition algorithm (RDA) and its parallel version for large-scale graphs. Section~\ref{sec-4} briefly discusses potential applications of our decomposition algorithm. Section~\ref{sec-5} outlines simulations conducted to validate our approach. Finally, Section~\ref{sec-6} presents our conclusions and offers a brief discussion.

\section{Preliminaries}\label{sec-2}
  
Throughout this paper, we focus exclusively on simple and undirected graphs which are connected and finite. A graph is denoted by $G = (V, E)$, comprising a set of nodes $V$ and a set of undirected edges $E$.  We use 
$(u,v)$ to indicate that nodes $u$ and $v$ are connected by an edge in $G$; thus, $u$ and $v$ are adjacent. The set 
$N_G(v)$ denotes the neighbors of node $v$ in $G$, and we omit the subscript $G$ when context permits.

For convenience, we define the useful sets:
\begin{itemize}
    \item $N[v] = N(v)\cup \{v\}$; \item $N(A) = \cup_{w\in A}N(w)\setminus A$;
    \item $N[A] = N(A)\cup A$.
\end{itemize}

An undirected path of length $k$ connecting $u,v$ in $G$, denoted by $l_{uv}$, is defined as a sequences of vertices $\ (u = c_0,c_1,\dots,c_k = v)$ such that each consecutive pair $c_{i}$ and $c_{i+1}$ are adjacent in $G$ for $i\in [0,k-1]\triangleq \{0,1,\ldots,k-1\}$ and all nodes $c_i$ and $c_j$ are distinct. If $u=v$, $l_{uv}$ is referred to as a cycle of $G$. An edge connecting two non-consecutive nodes in $l_{uv}$ is termed a chord of it. The internal nodes of the path $l_{uv}$  are denoted as $V^o(l_{uv})$.

A subset $A$ is termed complete in $G$ if every pair of nodes in $A$ is adjacent. For a subset $A \subseteq V$, the induced subgraph of $G$ on $A$ is denoted as $G_A=\left(A, E_A\right)$, where $E_A=E \cap(A \times A)$. We say that the subgraph $G_A$ is connected if, for any pair of nodes $x$ and $y$ in $A$, there exists a path $l_{xy}$ in $G_A$.
A connected component of $G$ is a subgraph $G_A$ that is connected, and for any $G_H \ subsetneq G_A, G_H$ is not connected in $G$.

For $u, v \in V$, a subset $S \subseteq V$ is called a $ uv$-separator of $G$ if every path connecting $u$ and $v$ must pass through $S$. If $S$ is a $uv$-separator and there exists no subset $S^{\prime} \subsetneq$ $S$ that is also a $uv$-separator of $G$, then $S$ is referred to as a minimal $uv$-separator. It is shown that $S$ is a minimal $xy$-separator if and only if there exist two connected components $X$ and $Y$ in $G_{V\backslash S}$ such that $N_G(X)=S=N_G(Y)$ \cite{leimer1993optimal}. In this case, $X$ and $Y$ are also referred to as full components of $S$. 

We say that a minimal separator of $G$ is a clique minimal separator of $G$ if it is complete. An atom of $G$ is defined as a maximal connected subgraph that contains no clique minimal separator. A clique minimal separator $S$ can decompose the graph $G$ into subsets $A \cup S$ and $B \cup S$, denoted as $(A, B, S)_G$, where $A$ is a full component of $S$, and $B=V \backslash(A \cup S)$, see \cite{berry2010introduction} for details.

For a graph $G$, a node ordering, denoted by $\alpha$, is a bijection from $V$ to $\{1, \ldots, n\}$. The label of node $v$ in ordering $\alpha$ is $\alpha(v)$. Let $v_1, v_2, \ldots, v_n$ be an ordering of $V$ and  $\mathcal{L}_i=\left\{v_i, v_{i+1}, \ldots, v_n\right\}$.
A graph $G$ is chordal if every cycle of length greater than 3 contains at least one chord. The Maximum Cardinality Search algorithm (MCS) was originally proposed by Tarjan \cite{tarjan1976maximum} and further developed by Tarjan and Yannakakis \cite{tarjan1984simple} for detecting chordal graphs. As the decomposition algorithm proposed in this paper relies on the MCS algorithm in arbitrary graphs, we emphasize and describe it in Algorithm \ref{alg-1}. In this work, we use the MCS algorithm to obtain a node ordering of $G$, which helps in locating the atoms of $G$. For simplicity, the node ordering of $G$ returned by the MCS algorithm is referred to as the MCS ordering of $G$.

To proceed with our work, we introduce the symbols $u^{-}$ and $u^{+}$ to denote specific time stamps. Here, $u^{-}$ represents the time at which vertex $u$ is selected to receive its number, and $u^{+}$ represents the time when vertex $u$ completes its numbering. Specifically, at time $u^{-}$, the weight of $v$ is denoted as $w_{u^{-}}(v)$, and at time $u^{+}$, it is denoted as $w_{u^{+}}(v)$. 
Similarly, for any subset $A \subseteq V$, $w_{u^{-}}(A)$ and $w_{u^{+}}(A)$ represent the highest weight among the unnumbered vertices in $A$ at times $u^{-}$ and $u^{+}$, respectively. See Figure~\ref{fig-3} for an illustration.

\begin{example}
	Let $G$ be the graph shown in ~Figure~\ref{fig-3}. To illustrate the process of labelling each node in this graph by using the MCS Algorithm~\ref{alg-1}, for each node $v$, the final label of $v$, $\alpha(v)$,  is written outside the parentheses next to it, while the weight $w(v)$ is written within the parentheses. It can be verified that $w_{c^-}(b) = 1$ and $w_{c^+}(b) = 2$, as $b$ is in the neighborhood of $c$ and remains unnumbered.  Similarly, we can find that $w_{c^-}(\{a,b\}) = \max\{w_{c^-}(a),w_{c^-}(b)\} = w_{c^-}(b)=1$ and $w_{c^+}(\{a,b\}) = \max\{w_{c^+}(a),\\w_{c^+}(b) \}= w_{c^+}(b) = 2$.

\end{example}

\begin{algorithm}[H]
	\caption{Maximum Cardinality Search Algorithm (MCS)}
	\label{alg-1}
	\begin{algorithmic}[1]
		\REQUIRE An arbitrary undirected graph $G=(V,E)$.
		\ENSURE A MCS ordering $\alpha$ of the nodes of $G$.
		\STATE Initialize: $\mathcal{L}_{n+1} = \emptyset$, $w(v) = 0$ for all $v\in V$;
		\FOR{$i$ from $n$ to 1}
		\STATE Choose a node $v\in V\backslash\mathcal{L}_{i+1}$ such that $w(v)$ is maximum;
		\STATE Update $\alpha(v)$ to $i$, and $\mathcal{L}_i$ to $\mathcal{L}_{i+1}\cup \{v\}$;
		\FOR{$v^\prime \in N(v)\cap (V\backslash \mathcal{L}_i)$}
		\STATE Update $w(v^\prime)$ to $w(v^\prime)+1$;
		\ENDFOR
		\ENDFOR
		\RETURN The ordering $\alpha$.
	\end{algorithmic}
\end{algorithm}

	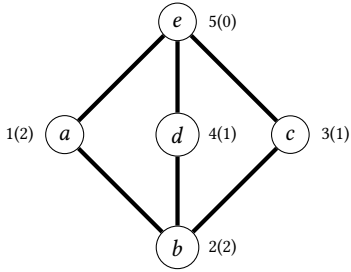
\begin{figure}[ht]
		\vskip 0.08in
		\begin{center}
			\begin{tikzpicture}[scale=0.75, every node/.style={minimum size=0.5cm}]
				\node[shape=circle, draw=black] (1) at (0,0){$a$};
				\node at (-0.8,0){\footnotesize 1(2)};
				\node[shape=circle, draw=black] (2) at (2,-2){$b$};
				\node at (2.8,-2){\footnotesize 2(2)};
				\node[shape=circle, draw=black] (3) at (4,0){$c$};
				\node at (4.8,0){\footnotesize 3(1)};
				\node[shape=circle, draw=black] (4) at (2,0){$d$};
				\node at (2.8,0){\footnotesize 4(1)};
				\node[shape=circle, draw=black] (5) at (2,2){$e$};
				\node at (2.8,2){\footnotesize 5(0)};
				\draw[-,ultra thick] (1)--(2);
				\draw[-,ultra thick] (4)--(5);
				\draw[-,ultra thick] (2)--(4);
				\draw[-,ultra thick] (1)--(5);
				\draw[-,ultra thick] (2)--(3);
				\draw[-,ultra thick] (3)--(5);
			
			\end{tikzpicture}		
			\caption{The labels and weights for vertices in the graph when applying the MCS algorithm.}
			\label{fig-3}
		\end{center}
		\vskip -0.2in
	\end{figure}

Before we proceed with our work, it's important to discuss the concept of convex hulls in undirected graphs. A subgraph $G_A$ is considered convex in $G$ if, for any non-adjacent nodes $u$ and $v$ in $A$, there is no path $l_{uv}$ such that $V^o\left(l_{uv}\right) \subseteq V\backslash A$.
For any subset $R \subseteq V$, it is known that there exists a unique minimal convex subgraph $G_H$ in $G$ containing $R$, such that for any $R \subseteq B \subsetneq H, G_B$ is not convex in $G$. In this case, $G_H$ (or simply $H$ ) is referred to as the convex hull containing $R$. Convex hull finding algorithms are extensively studied in the statistics community due to their applications in model collapsibility.

Madigan and Mosurski \cite{madigan1990extension} explored identifying convex hulls on chordal graphs while seeking the minimum collapsible set for a decomposable graphical model. Wang et al. \cite{wang2011finding} expanded this research to general graphs by introducing an algorithm that utilizes clique minimal separator decomposition to identify convex hulls in any undirected graph. Building on the importance of graphical models, Heng and Sun \cite{heng2023algorithms} delved deeper into the issue and introduced efficient algorithms aimed at expanding convex hulls in undirected graphs.

In this work, we present a recursive decomposition algorithm inspired by the Close Minimal Separator Absorption (CMSA) algorithm proposed by \cite{heng2024algorithms}. This algorithm takes an undirected graph $G$ and a subset $R$ as input, and it computes the convex hull containing $R$ by absorbing specific minimal separators within 
$G$. 

\section{An efficient recursive decomposition algorithm}\label{sec-3}

Berry et al. \cite{berry2004maximum} studied the node labeled one (i.e., the node with the smallest label) by the Maximum Cardinality Search (MCS) algorithm in an arbitrary undirected graph G. They showed that this node can be given priority for removal in specific minimal elimination sequences. The proposition we aim to prove asserts that the node with the smallest number in the MCS ordering and its neighboring nodes are part of a distinct component in $G$.

\begin{proposition}\label{prop-1}
	Suppose that $G=(V, E)$ is an undirected graph with a node ordering $\alpha$ by the MCS algorithm. Let $v:=\underset{w\in V}{\arg\min}\ \alpha(w)$. The convex hull $H$ that contains the set $N_G[v]$ in $G$ forms an atom of $G$.
\end{proposition}
\begin{proof}
	To begin with, we show that $G_H$ contains no clique minimal separator. Suppose there exists a clique minimal separator $S$ in $H$ and a corresponding decomposition $(A, B, S)_{G}$ such that $A\cap H,B\cap H\neq \emptyset$. The following two cases are to be considered:
	
	(1) For $v\notin S$, without loss of generality, we assume $v\in A$. It is obvious that $N_G[v]\subseteq A\cup S$. Since $A\cup S$ is convex in $G$, which contradicts the fact that $H$ is the convex hull containing $N_G[v]$.
	
	(2) For $v\in S$, it is evident that $N_G(v)\cap A$ and $N_G(v)\cap B$ are both non-empty. Otherwise, it would result in the same contradiction. Without loss of generality, let's assume $a_1 = \underset{w\in V}{\arg\max}\ \alpha(w)$ and $a_1\in A$, and let $b_1 = \underset{w\in B}{\arg\max} \{ \alpha(w) < \alpha(a_1) \}$.  Before numbering $b_1$, the nodes in $N_G(v)\cap A$ are unnumbered. Otherwise, if they were numbered, $w_{b_1^-}(v)>w_{b_1^-}(B)\geq w_{b_1^-}(b_1)$, which contradicts choosing $b_1$ as the next node to be numbered. Furthermore, $w_{b_1^-}(v) = w_{b_1^-}(b_1) \geq w_{b_1^-}(A)$. Similarly, let $a_2 = \underset{w \in A}{\arg\max} \{ \alpha(w) < \alpha(b_1) \}
	$. Before numbering $a_2$, the nodes in $N_G(v)\cap B$ are unnumbered. If any of these nodes were numbered, it would lead to $w_{a_2^-}(v)>w_{a_2^-}(A)$, which contradicts choosing $a_2$ as the next node to be numbered. Moreover, as $a_2$ is selected for numbering in this step, it follows that $w_{a_2^-}(v) = w_{a_2^-}(a_2) \geq w_{a_2^-}(B)$. In other words, the nodes in $N_G(v)\cap A$ and $N_G(v)\cap B$ cannot receive numbering before $v$, leading to a contradiction.
	
	The fact that $G_H$ is the atom of $G$ is evident because for any connected subgraph $G_A\supsetneq G_H$, there exists a clique minimal separator $S^\prime \subseteq N_G(V\backslash H)$ in $G_A$.
\end{proof}

Proposition \ref{prop-1} shows that an atom of graph $G$ can be located with the aid of the node ordering returned by applying the MCS algorithm. In the following proposition, we will demonstrate that all atoms of $G$ can be successively acquired using the MCS ordering.

\begin{proposition}\label{prop-2}
	Suppose that $G=(V, E)$ is an undirected graph with an MCS ordering $\alpha$. Let $H_v$ denote the convex hull containing the smallest numbered  node $v$ and its neighborhood $N_G(v)$, and $V^\prime =  N_G[V \backslash H_v]$, $u =  \underset{x\in V^\prime}{\arg\min}\ \alpha(x)$. Then the convex hull $H_u$ containing $N_{G^\prime}[u]$ in $G^\prime$, is an atom of $G^\prime$, where $G^\prime =G_{V^\prime}$. Furthermore, if $H_u \nsubseteq H_v$, then $G^\prime_{H_u}$ is also an atom of $G$.
\end{proposition}
\begin{proof}
	Suppose there exists a clique minimal separator $S$ in $H_u$ and a corresponding decomposition $(A, B, S)_{G^\prime}$ such that $A\cap H_u,B\cap H_u\neq \emptyset$. By similar analysis, we have that $G^\prime_{H_u}$ is an atom of $G^\prime$.
	
	If $H_u \nsubseteq H_v$, we prove that $G^\prime_{H_u}$ is also an atom of $G$. Otherwise, there exists an atom $G_{H_1}$ in $G$ such that $H_u\subsetneq H_1$, and it is evident that $H_1 \subseteq V^\prime$. Thus  $G_{H_1}$ is also an atom in $G^\prime$ due to the hereditary property of atoms. This contradicts the fact that $G^\prime_{H_u}$ is an atom in $G^\prime$. 
\end{proof}

The propositions \ref{prop-1} and \ref{prop-2} serve as the theoretical foundation for our recursive decomposition algorithm, designed to identify all atoms within a graph. The algorithm operates as follows:

\begin{itemize}
    \item Initially, obtain a node ordering of $G$ using the MCS algorithm.
    \item Next, choose the node $v$ with the smallest number from the remaining set of nodes $V'$ (as defined in Proposition~\ref{prop-2}).
    \item Then, utilize the CMSA algorithm to determine the convex hull containing $v$ and its neighbors.
    \item Repeat the second and third steps until all nodes have been visited.
\end{itemize}

For the pseudocode of this algorithm, refer to Algorithm~\ref{alg-2} for a detailed explanation.

\begin{theorem}
	The RDA algorithm precisely returns all atoms of the graph $G$.
\end{theorem}
\begin{proof}
	By employing the argument of induction, it follows directly from Propositions \ref{prop-1} and \ref{prop-2}. For brevity, we omit the proof here.
\end{proof}

\begin{algorithm}[H]
	\caption{Recursive Decomposition Algorithm (RDA)}
	\label{alg-2}
	\begin{algorithmic}[1]
		\REQUIRE An undirected graph $G=(V,E)$ and its MCS ordering $\alpha$.
		\ENSURE All atoms of $G$.
		\STATE $\mathbf{procedure}~~\mathrm{RDA}(G,\alpha)$
		\STATE Initialize: $\mathcal{U}_G = \emptyset$;
		\REPEAT
		\STATE $v = \underset{w\in V}{\arg\min}\ \alpha(w)$;
		\STATE $H = \mathrm{CMSA}(G, N_G[v])$;
		\STATE Update $V$ to $N_G[V\backslash H]$, and $G$ to $G_V$;
		\IF{$H\nsubseteq U$ for all $U\in \mathcal{U}_G$}
		\STATE Add $H$ to $\mathcal{U}_G$;
		\ENDIF
		\UNTIL{$V=\emptyset$};
		\RETURN $\mathcal{U}_G$
	\end{algorithmic}
\end{algorithm}

\begin{proposition}
	The complexity of the RDA algorithm is at most $O(nm)$, where $n$ is the number of nodes, and $m$ is the number of edges in $G$.
\end{proposition}
\begin{proof}
	In the worst-case scenario, the algorithm only has one atom, which can be found using the CMSA algorithm in $O(nm)$ time.
\end{proof}

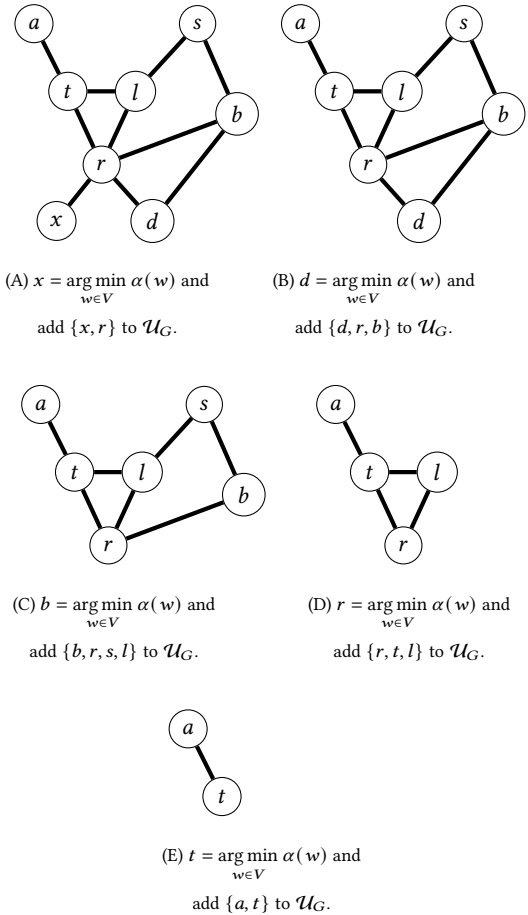
\begin{figure}[ht]
	\vskip 0.08in
	\begin{center}
		\begin{tikzpicture}[scale=0.75]
			\node[shape=circle, draw=black] (1) at (0.2,0){$a$};
			\node[shape=circle, draw=black] (2) at (0.8,-1.2){$t$};
			\node[shape=circle, draw=black] (3) at (2,-1.2){$l$};
			\node[shape=circle, draw=black] (4) at (3.1,0){$s$};
			\node[shape=circle, draw=black] (5) at (3.8,-1.6){$b$};
			\node[shape=circle, draw=black] (6) at (1.4,-2.5){$r$};
			\node[shape=circle, draw=black] (7) at (0.6,-3.5){$x$};
			\node[shape=circle, draw=black] (8) at (2.3,-3.5){$d$};
			\draw[-,ultra thick] (1)--(2);
			\draw[-,ultra thick] (2)--(3);
			\draw[-,ultra thick] (2)--(6);
			\draw[-,ultra thick] (3)--(4);
			\draw[-,ultra thick] (3)--(6);
			\draw[-,ultra thick] (6)--(7);
			\draw[-,ultra thick] (8)--(6);	
			\draw[-,ultra thick] (5)--(8);	
			\draw[-,ultra thick] (6)--(5);	
			\draw[-,ultra thick] (4)--(5);
			\node at (1.5,-4.7){\footnotesize (A) $x = \underset{w\in V}{\arg\min}\ \alpha(w)$ and };
			\node at (1.5,-5.4){\footnotesize add $\{x,r\}$ to $\mathcal{U}_G$.};	
		\end{tikzpicture}
        \vspace{0.5cm}
		\begin{tikzpicture}[scale=0.75]
			\node[shape=circle, draw=black] (1) at (0.2,0){$a$};
			\node[shape=circle, draw=black] (2) at (0.8,-1.2){$t$};
			\node[shape=circle, draw=black] (3) at (2,-1.2){$l$};
			\node[shape=circle, draw=black] (4) at (3.1,0){$s$};
			\node[shape=circle, draw=black] (5) at (3.8,-1.6){$b$};
			\node[shape=circle, draw=black] (6) at (1.4,-2.5){$r$};
			\node[shape=circle, draw=black] (8) at (2.3,-3.5){$d$};
			\draw[-,ultra thick] (1)--(2);
			\draw[-,ultra thick] (2)--(3);
			\draw[-,ultra thick] (2)--(6);
			\draw[-,ultra thick] (3)--(4);
			\draw[-,ultra thick] (3)--(6);
			\draw[-,ultra thick] (8)--(6);	
			\draw[-,ultra thick] (5)--(8);	
			\draw[-,ultra thick] (6)--(5);	
			\draw[-,ultra thick] (4)--(5);	
			\node at (1.5,-4.7){\footnotesize (B) $d = \underset{w\in V}{\arg\min}\ \alpha(w)$ and };
			\node at (1.5,-5.4){\footnotesize add $\{d,r,b\}$ to $\mathcal{U}_G$.};			
		\end{tikzpicture}
		\vspace{0.5cm}
		\begin{tikzpicture}[scale=0.75]
			\node[shape=circle, draw=black] (1) at (0.2,0){$a$};
			\node[shape=circle, draw=black] (2) at (0.8,-1.2){$t$};
			\node[shape=circle, draw=black] (3) at (2,-1.2){$l$};
			\node[shape=circle, draw=black] (4) at (3.1,0){$s$};
			\node[shape=circle, draw=black] (5) at (3.8,-1.6){$b$};
			\node[shape=circle, draw=black] (6) at (1.4,-2.5){$r$};
			\draw[-,ultra thick] (1)--(2);
			\draw[-,ultra thick] (2)--(3);
			\draw[-,ultra thick] (2)--(6);
			\draw[-,ultra thick] (3)--(4);
			\draw[-,ultra thick] (3)--(6);
			\draw[-,ultra thick] (6)--(5);	
			\draw[-,ultra thick] (4)--(5);	
			\node at (1.5,-3.7){\footnotesize (C) $b = \underset{w\in V}{\arg\min}\ \alpha(w)$ and };
			\node at (1.5,-4.4){\footnotesize add $\{b,r,s,l\}$ to $\mathcal{U}_G$.};			
		\end{tikzpicture}		
		\hspace{0.3cm}
        \begin{tikzpicture}[scale=0.75]
			\node[shape=circle, draw=black] (1) at (0.2,0){$a$};
			\node[shape=circle, draw=black] (2) at (0.8,-1.2){$t$};
			\node[shape=circle, draw=black] (3) at (2,-1.2){$l$};
			\node[shape=circle, draw=black] (6) at (1.4,-2.5){$r$};
			\draw[-,ultra thick] (1)--(2);
			\draw[-,ultra thick] (2)--(3);
			\draw[-,ultra thick] (2)--(6);
			\draw[-,ultra thick] (3)--(6);
			\node at (1.5,-3.7){\footnotesize (D) $r = \underset{w\in V}{\arg\min}\ \alpha(w)$ and };
			\node at (1.5,-4.4){\footnotesize add $\{r,t,l\}$ to $\mathcal{U}_G$.};			
		\end{tikzpicture}
		\hspace{0.3cm}
		\begin{tikzpicture}[scale=0.75]
			\node[shape=circle, draw=black] (1) at (0.2,0){$a$};
			\node[shape=circle, draw=black] (2) at (0.8,-1.2){$t$};
			\draw[-,ultra thick] (1)--(2);
			\node at (1.5,-2.4){\footnotesize (E) $t = \underset{w\in V}{\arg\min}\ \alpha(w)$ and };
			\node at (1.5,-3.1){\footnotesize add $\{a,t\}$ to $\mathcal{U}_G$.};			
		\end{tikzpicture}
		\caption{The process of decomposing the undirected graph $G$ by applying the RDA algorithm. }
		\label{fig-1}
	\end{center}
	\vskip -0.2in
\end{figure}

To give an illustration of the RDA algorithm, we provide an example to clarify its execution process.

\begin{example}\label{exm-1}	
 Let $G$ be an undirected graph shown in Figure~\ref{fig-1}(A). Applying the MCS algorithm yields the ordering $\alpha=\{x,d,b,s,r,l,t,\\a\}$. With the node ordering in hand, we apply the algorithm RDA to $G$, and the process of recursive decomposition is listed as follows:
\begin{enumerate}[(i)]
    \item First Iteration:
    \begin{itemize}
        \item $x=\underset{w \in V}{\arg \min } \alpha(w)$.
\item Convex hull $H$ containing $N_G[x]=\{x,r\}$ is $\{x,r\}$.
\item Update: $V=V \backslash\{x\}, G=G_V$.
\item  Add $\{x,r\}$ to $\mathcal{U}_G$.
    \end{itemize}
    \item Second Iteration:
    \begin{itemize}
        \item $d=\underset{w \in V}{\arg \min } \alpha(w)$.
        \item Convex hull $H$ containing $N_G[d]=\{d,r,b\}$ is $\{d,r,b\}$.
        \item  Update: $V=V \backslash\{d\}, G=G_V$.
       \item  Add $\{d,r,b\}$ to $\mathcal{U}_G$.
    \end{itemize}
\item Third Iteration:
\begin{itemize}
    \item $b=\underset{w \in V}{\arg \min } \alpha(w)$.
    \item Convex hull $H$ containing $N_G[b]=\{b,r,s\}$ is $\{b,r,s, l\}$.
\item Update: $V=V \backslash\{b,s\}, G=G_V$.
\item  Add $\{b,r,s, l\}$ to $\mathcal{U}_G$.
\end{itemize}
\item Fourth Iteration:
\begin{itemize}
    \item $r=\underset{w \in V}{\arg \min } \alpha(w)$.
    \item Convex hull $H$ containing $N_G[r]=\{r,t,l\}$ is $\{r,t,l\}$.
    \item  Update: $V=V \backslash\{r,l\}, G=G_V$.
    \item Add $\{r,t,l\}$ to $\mathcal{U}_G$.
\end{itemize}
\item Fifth Iteration:
\begin{itemize}
    \item $t=\underset{w \in V}{\arg \min } \alpha(w)$.
    \item Convex hull $H$ containing $N_G[t]=\{a,t\}$ is $\{a,t\}$.
    \item Update: $V=\emptyset$.
    \item  Add $\{a,t\}$ to $\mathcal{U}_G$.
    \end{itemize}
\end{enumerate}
After the above iterations, the RDA algorithm outputs $\mathcal{U}_G$, which consists of the sets added in each iteration, i.e.,
$$
\mathcal{U}_G=\{\{x,r\},\{d,r,b\},\{b,r,s, l\},\{r,t,l\},\{a,t\}\}.
$$

\end{example}

From Example~\ref{exm-1}, it is not difficult to see that the proposed RDA can be used to find clique minimal separators. We summarize this fact as a corollary.

\begin{corollary}\label{cor-lem-1}
	Suppose that $G=(V, E)$ is an undirected graph and $G_H$ is an atom of $G$. Then $N_G(M)$ is a clique minimal separator of $G$, where $M$ is a connected component of $G_{V\backslash H}$.
\end{corollary}

\begin{proof}
	It is obvious that $N_G(M)$ is a minimal $xy$-separator and is complete for any $x\in H\backslash N_G(M), y\in M$.
\end{proof}
Corollary~\ref{cor-lem-1} indicates that the RDA algorithm can be used to identify the clique minimal separators of the graph at each iteration. In previous research, clique minimal separators were primarily used to discover the atoms of the graph. However, since we have already obtained all atoms of the graph, this paper omits the step of searching for clique minimal separators.

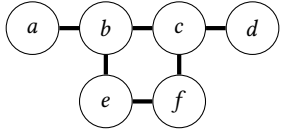
\begin{figure}[ht]
	\vskip 0.08in
	\begin{center}
		\begin{tikzpicture}[scale=0.75, every node/.style={minimum size=0.7cm}]
			\node[shape=circle, draw=black] (1) at (0,0){$a$};
			\node[shape=circle, draw=black] (2) at (1.3,0){$b$};
			\node[shape=circle, draw=black] (3) at (2.6,0){$c$};
			\node[shape=circle, draw=black] (4) at (3.9,0){$d$};
			\node[shape=circle, draw=black] (5) at (1.3,-1.3){$e$};
			\node[shape=circle, draw=black] (6) at (2.6,-1.3){$f$};
			\draw[-,ultra thick] (1)--(2);
			\draw[-,ultra thick] (2)--(3);
			\draw[-,ultra thick] (3)--(4);
			\draw[-,ultra thick] (2)--(5);
			\draw[-,ultra thick] (5)--(6);
			\draw[-,ultra thick] (6)--(3);				
		\end{tikzpicture}		
		\caption{The undirected graph $G$ and its MCS ordering $\alpha=\{e,f,d,c,b, a\}$}
		\label{fig-2}
	\end{center}
	\vskip -0.2in
\end{figure}

Taking into account the undirected graph $G$ in ~Figure~\ref{fig-2}, where $H = \{b,c,e,f\}$ is the first atom identified by the RDA algorithm. We can see that RDA can recursively search for atoms in the subgraphs $G_{\{a,b\}}$ and $G_{\{c,d\}}$ in the second iteration, respectively. Therefore, we propose the following parallel version of the RDA algorithm:

\begin{algorithm}[H]
	\caption{Parallel Recursive Decomposition Algorithm. (PRDA)}
	\label{alg-3}
	\begin{algorithmic}[1]
		\REQUIRE An undirected graph $G=(V,E)$ and its MCS ordering $\alpha$.
		\ENSURE All atoms of $G$.
		\STATE $\mathbf{procedure}~~\mathrm{PRDA}(G,\alpha)$;
		\STATE Initialize: $\mathcal{U}_G = \emptyset $;
		\STATE $v = \underset{w\in V}{\arg\min}\ \alpha(w)$;
		\STATE $H = CMSA(G, N_G[v])$;
		\STATE Add $H$ to $\mathcal{U}_G$;
		\STATE \textbf{for} each connected component $M_i$ of $G_{V\backslash H}$ \textbf{do in parallel}
		\STATE ~~~~~Update $V_i$ to $N_G[M_i]$;
		\STATE ~~~~~$\mathrm{PRDA}(G_{V_i},\alpha_{V_i})$;
		\RETURN $\mathcal{U}_G$
	\end{algorithmic}
\end{algorithm}

The PRDA algorithm may not always be able to perform parallel searches. For example, in ~Figure~\ref{fig-1}, subgraph $G_{V\backslash H}$ always has one connected component. This means that, in practice, the PRDA algorithm is not always more efficient than RDA, and sometimes they are even less efficient because of the increased cost of parallelism.

\section{Experiments}\label{sec-5}

We compared the performance of the RAD algorithm with traditional decomposition algorithms in a Python environment. All our experiments were run on a system with an Intel(R) Xeon(R) Silver 4215R CPU @ 3.20GHz (3.19GHz base, 2 processors) and 128GB of RAM.

In general, a random graph does not contain clique minimal separators, which means that the graph itself is an atom. Therefore, we compared the efficiency of the RDA algorithm with the Xu and Guo algorithm \cite{xu2012new} in real networks\footnote{These networks are sourced from repositories https://snap. stanford.edu/data/index.html and https://networkrepository.com/.}. The algorithm of Xu and Guo combines Leimer's algorithm with the MCS-M algorithm \cite{berry2004maximum}, which has been experimentally demonstrated to be efficient for decomposition. The experimental setup is as follows: 

\begin{itemize}
	\item We apply both algorithms to decompose each real network and record the execution time for each run.
	\item  This process is repeated 20 times to compute the average runtime for each algorithm.
\end{itemize}
\begin{table}[]
	\caption{The average runtime (in seconds) for the two decomposition algorithms.}
	\label{tab-1}
		\begin{tabular}{ccccc}
			\hline
			Network                        & nodes     & edges    & $T_1$     & $T_2$ \\ \hline
			Animal Network-1 			& 103   & 151    & 0.03   & 0.15       \\
			Animal Network-2            & 445   & 1423   & 0.14   & 0.94       \\
			bio-CE-GT                      & 924   & 3239   & 0.5    & 9.75       \\
			bio-CE-GN                      & 2200  & 53683  & 1.11   & 193.43     \\
			bio-DR-CX                      & 3289  & 84940  & 2.09   & 892.38     \\
			as20000102                     & 6474  & 13895  & 20.34  & ---      \\
			CA-HepTh                       & 9877  & 25998  & 184.99 & ---       \\
			CA-CondMat                     & 23133 & 93497  & 857.99 & ---          \\ 
			Email-Enron                    & 36692 & 183831 & 859.77 & ---        \\ \hline
		\end{tabular}%
\end{table}
Table \ref{tab-1} displays the average computation times for two algorithms: $T_1$ for the RDA algorithm and $T_2$
for the algorithm proposed by Xu and Guo \cite{xu2012new}. Entries marked with `---' indicate that the decomposition time of the algorithm exceeded acceptable limits, resulting in the exclusion of those results. The experiments demonstrate that the RDA algorithm significantly outperforms the traditional algorithm in terms of efficiency.

\section{Applications}\label{sec-4}

Research suggests that solving the problem can be tackled independently in components, which are then integrated to achieve a global solution. However, the most effective current method for decomposing a graph into components involves extracting clique minimal separators from minimal triangulations. This can result in decomposition times that rival or exceed direct computation costs. Specifically, the time required for minimal clique decomposition algorithms in high-dimensional graphs is substantial. As a result, previous decomposition efforts have primarily remained theoretical, lacking practical tools for user-friendly divide-and-conquer approaches.

Our work addresses several key challenges:

Each component undergoes an independent minimal (or minimum) triangulation, with the resultant edge sets combined to form the minimal (or minimum) triangulation of G. Existing clique separator decomposition algorithms are often unsuitable for directly constructing this triangulation. Our approach integrates minimal triangulation during the decomposition process, allowing previously identified components to be triangulated while searching for subsequent ones, thereby enhancing efficiency. Similar methods can be adapted for identifying maximal cliques, determining treewidth, and graph coloring.

The maximum prime subgraph, as proposed by Olesen and Madsen \cite{olesen2002maximal}, plays a vital role in the propagation of the Bayesian network as a structured tree. Our method effectively constructs the maximum prime decomposition of Bayesian networks.

In the context of DAG (Directed Acyclic Graph) structure learning, decomposition is pivotal. However, the challenge of minimal triangulation in high-dimensional graphs has hindered progress \cite{xie2006decomposition,xu2011structural}. Our algorithm effectively overcomes this challenge.
  
\section{Conclusion and Discussion}\label{sec-6}

Graph decomposition offers an efficient approach to handling complex undirected graphs by breaking them down into smaller, localized subproblems that can be integrated to yield overall solutions. Traditional decomposition algorithms often face challenges such as minimal triangulation and finding clique minimal separators, leading to significant computation costs for high-dimensional graphs.

This paper introduces a recursive decomposition algorithm that identifies convex hulls containing specific variables, effectively circumventing the bottlenecks encountered by traditional methods. Experimental results demonstrate the algorithm's substantial efficiency compared to conventional approaches, presenting it as a powerful tool for users. Additionally, we explore a parallel variant of the algorithm, whose performance is influenced by the selection of MCS order. Future research will focus on integrating specific MCS variants with the PRDA algorithm to further enhance its applicability.

\bibliographystyle{ACM-Reference-Format}
\bibliography{sample-base}



\end{document}